\documentclass[amsmath,amssymb,notitlepage]{revtex4-1}

\usepackage{enumerate}
\usepackage{latexsym}
\usepackage{amsfonts}
\usepackage{amsthm}
\usepackage{color}
\usepackage{bbm}
\usepackage{dsfont}
\usepackage{graphicx}
\usepackage{textcomp}
\usepackage{enumerate}
\usepackage{graphicx,tikz}
\usetikzlibrary{arrows,matrix,calc}
\usepackage[colorlinks=true,linkcolor=blue,citecolor=blue,urlcolor=blue]{hyperref}
\usepackage{mathrsfs} 
\usepackage{booktabs} 
\usepackage{array}

\newtheorem{theorem}{Theorem}
\newtheorem{proposition}[theorem]{Proposition}
\newtheorem{lemma}[theorem]{Lemma}

\DeclareMathOperator{\Tr}{Tr} 
\DeclareMathOperator{\reals}{\mathbb{R}}
\DeclareMathOperator{\comp}{\mathbb{C}}

\newcommand{\norm}[1]{\left\lVert#1\right\rVert}
\newcommand{\quasi}[1]{{\left\vert\kern-0.25ex\left\vert\kern-0.25ex\left\vert #1 
    \right\vert\kern-0.25ex\right\vert\kern-0.25ex\right\vert}}

\usepackage{geometry}
\begin{document}
\title{(Co-)type and the linear stability of Wigner's symmetry theorem}

\author{Javier Cuesta}
\email{j.cuesta@tum.de}
\affiliation{Department of Mathematics, Technische Universit\"at M\"unchen, 85748 Garching, Germany}
\affiliation{Munich Center for Quantum
Science and Technology (MCQST),  M\"unchen, Germany}
\date{\today}

\vspace{-1cm}\begin{abstract}

We study the relation between the linear stability of almost-symmetries and the geometry of the Banach spaces on which these transformations are defined. We show that any transformation between finite dimensional Banach spaces that preserves transition probabilities up to an additive error admits an approximation by a linear map, and the quality of the approximation depends on the type and cotype constants of the involved spaces. 
\end{abstract}

\maketitle

\tableofcontents

\section{ Introduction}

In the work of N. J. Kalton~\cite{K91,BK00,KP79} we can find novel ideas and methods for the stability of functional equations which depart from the classical methods of Hyers, Ulam and Rassias~\cite{Jung11}. In Ref.~\cite{K91} (see Theorem 2.2) Kalton 
provides a sharp bound on the stability of the additive map in $\reals^{n}$ for the so-called singular case. His proof makes use of probabilistic and geometric methods in Banach space theory. This paper ends with a sketch on how the theory of twisted sums in Banach space theory could be used to obtain the same result.  In this note, we study this last idea and use it to obtain a small improvement in the linear stability of Wigner's theorem~\cite{CW18}.\\

Wigner's celebrated symmetry theorem is not only central for physics, but it also finds an important role in many preservers problems. A \textit{preserver problem} deals with the characterization of maps, primarily on matrix spaces and operator algebras, that preserve certain functional, subset, or an invariant. In particular, in the field of Quantum Information Theory (QIT) it has being shown~\cite{MLN13} that the only mapping $T$ that preserves the $f-$divergences (this includes the von Neumann and relative entropy) is a Wigner symmetry transformation, i.e. of the form $T(x)=UxU^{*}$ where $U$ is either a unitary or antiunitary transformation on $\comp^{d}$. It turns out that most of the proofs of different preservers problems can be reduced to Wigner's theorem. Therefore, it is natural to expect that sharp bounds on the stability of Wigner's theorem could provide good approximations for a wide range of almost-preserving problems. It is worth pointing out that there exists a close relation between geometric functional analysis and many questions in QIT~\cite{ABmB}. This is the point of view that we want to motivate here. \\




Throughout this note, we will be entirely concerned with finite dimensional Banach spaces and the twisted sums generated by almost-linear maps. A map $F:X\to Y$ between Banach spaces will be called \textit{almost-linear} if it satifies the following two conditions: 

\begin{enumerate}[(i)]
 \item $F(\lambda x)=\lambda F(x)$ for all $\lambda\in\reals$ and $x\in X$,
 \item there exist a $\delta>0$ such that for any finite sequence $(x_{i})_{i=1}^{m}\subset X$, $m\in\mathbb{N}$ and $\lambda\in\reals^{m}$, 

\begin{equation}\label{eq:almostl}
\norm{\sum^{m}_{i=1}\lambda_{i}F(x_{i})-F\left(\sum^{m}_{i=1}\lambda_{i}x_{i}\right)}_{Y}\leq \delta \sum^{m}_{i=1}|\lambda_{i}|\norm{x_{i}}_{X}.
\end{equation}

\end{enumerate}

We will show that for every almost-linear map $F$ there exist a linear map $H$ whose distance to $F$ depends additively on $\delta$ and on some geometric invariants of the domain and target space of $F$ (see Theorem~\ref{th:k}). The Banach space numbers used to express the results are the type and cotype constants which we introduce now. Let $\{\gamma_{j}\}_{j=1}^{n}$ be a sequence of independent real Gaussian random variables, i.e. for each Borel subset $B\subset\reals$, each random variable has a distribution 

\begin{equation*}
\mu(\gamma \in B)=\frac{1}{(2\pi)^{1/2}}\int_{B}e^{-\frac{t^{2}}{2}}dt.
\end{equation*}

Let $X$ be a Banach space with norm $\norm{\cdot}$ and let $p\in[1,2]$, $q\in[2,\infty)$. For every positive interger $n$ we define $T_{p,n}(X),C_{q,n}(X)$ to be the smallest constants such that for arbitrary sequences $\{x_{j}\}_{j=1}^{n}\subset X$, we have

\begin{align*}
&\left(\mathbb{E}\norm{\sum_{j=1}^{n}\gamma_{j}x_{j}}^{2}\right)^{1/2}\leq T_{p,n}(X) \left(\sum_{j=1}^{n}\norm{x_{j}}^{p}\right)^{1/p}, \\
  C_{p,n}(X)^{-1} \left(\sum_{j=1}^{n}\norm{x_{j}}^{q}\right)^{1/q} \leq &\left(\mathbb{E}\norm{\sum_{j=1}^{n}\gamma_{j}x_{j}}^{2}\right)^{1/2}.
\end{align*} 

The space $X$ is said to be of Gaussian type $p$ (resp. Gaussian cotype $q$) if $T_{p}(X)=\sup_{n}T_{p,n}(X)<\infty$ (resp. $C_{q}(X)=\sup_{n}C_{q,n}(X)<\infty$). One can analogously define the Rademacher type and cotype by exchanging the Gaussian sequence by a Rademacher sequence. The results shown in this note are valid for both notions of type and cotype.\\

For $r\in [1,\infty)$ we denote by $S^{d}_{r}$ the Hermitian part of the  $d-$dimensional $r-$Schatten class and by $l_{r}^{d}$ the classical space of $r-$summable sequences in $\reals^{d}$; the space $S^{d}_{r}$ is a real Banach space with norm $\norm{x}_{r}:=\left(\Tr |x|^{r}\right)^{1/r} $. Table~\ref{tab:type} summarizes the behaviour of the type and cotype constants for the $r-$Schatten classes that we use. \\

\begin{table}[h]
\centering
    \begin{tabular}{  p{3cm}  p{3cm} p{3cm}}
     & \textbf{Type} $p\in[1,2]$ & \textbf{Cotype} $q\in[2,\infty]$\\ 
    $l_{1}^{d}$ & $d^{1-\frac{1}{p}}$ & $\sqrt{2}$   \\
    \text{Hilbert space} & 1 & 1  \\ 
    $l_{\infty}^{d}$ & $\backsim(\log d)^{1-1/p}$ & $d^{1/q}$   \\ 
    $S_{1}^{d}$ & $d^{1-1/p}$ & $\sqrt{e}$   \\  
    $S_{\infty}^{d}$ & $(4\log d)^{1-1/p}$ & $d^{1/q}$   \\
    \end{tabular}
      \caption{Upper bounds for the Rademacher type and cotype constants of the spaces $l^{d}_{r}$ and $S^{d}_{r}$. The Gaussian type and cotype for these spaces behave in the same way, up to a factor of $\sqrt{2/\pi}$, as the Rademacher type and cotype. For a Hilbert space the type and cotype constants are always equal to one. }
  \label{tab:type}
\end{table}


We now introduce some notation. The set of rank-one projections in $\comp^{d\times d}$ is denoted by $\mathbb{P}(\comp^{d})$. The unit ball of a space $Z$ is written as $B_{Z}$. The convex hull of a set $S$ is the set of convex combinations of elements of $S$, which we denote by $\operatorname{conv}(S)$. The set of linear maps between $X$ and $Y$ is $L(X,Y)$. A linear projection $P\in L(X,Y)$ is a linear map such that $P^{2}=P$. Finally, we denote by $\langle x,y\rangle:=\Tr xy$ the Hilbert-Schmidt inner product in the real vector space of Hermitian matrices $\mathcal{H}_{d}$.\\

In the next section, we introduce a special space which will generate the linear approximation to the almost-linear map $F:X\to Y$. This space is an extension of $X$ and $Y$ and is called a twisted sum (basically because it ``twists" the unit ball of $X$ and $Y$ according to $F$). Twisted sums were extensively studied by Kalton~\cite{KP79} in the context of the three-space problem. In particular, Kalton showed that twisted sums are in correspondence with quasi-linear maps; this is a weaker condition than almost-linearity, but for our purposes it suffices to say that any almost-linear map is a quasi-linear map. See Ref.~\cite{CG97} for a detailed exposition of this topic.

\section{Finite dimensional twisted sums}

Let $X,Y$ be two Banach spaces with dimension $d_{1},d_{2}$, respectively. The twisted sum of $Y$ and $X$ is a $(d_{1}+d_{2})-$dimensional space $Z$ that contains a subspace $Y_{0}$ which is isomorphic to $Y$ and such that $Z/Y_{0}$ is isomorphic to $X$. The twisted sums which interes us are constructed with an almost-linear function $F$. Consider $\delta>0$ and the product $Y\oplus X$ (the order is important) endowed with the quasi-norm:

\begin{equation}\label{eq:qnorm}
\quasi{(y,x)}_{F}:= \frac{\norm{y-F(x)}_{Y}}{\delta}+\norm{x}_{X}.
\end{equation}

Then $Y_{0}=\{(y,0): y\in Y\}$ is $\delta^{-1}-$isometric to $Y$ and $Z/Y_{0}$ isometric to $X$. Note that since $F$ is homogeneous, $\quasi{(-y,-x)}_{F}=\quasi{(y,x)}_{F}$ and $\quasi{(y,x)}_{F}=0$ implies $(y,x)=0$. Although $\quasi{(y_{1},x_{1})+(y_{2},x_{2})}_{F}\leq 2 (\quasi{(y_{1},x_{1})}_{F}+\quasi{(y_{2},x_{2})}_{F})$ we can still endow $Z$ with a norm. The twisted sum $Z$ can be made into a Banach space with the norm 

\begin{equation}\label{eq:benvelope}
\norm{(y,x)}:=\inf \left\{\sum_{j}\quasi{(y_{j},x_{j})}_{F}: (y,x)=\sum_{j}(y_{j},x_{j}) \right\}.
\end{equation}

The fact that the above expression defines a norm will be shown below. The completion of a quasi-Banach space $Z$ whose dual is non-trivial with respect to this norm is known as the Banach envelope of $Z$~\cite{KPR84}. In order to avoid charged notation, we also denote the Banach envelope by $Z$.

\begin{lemma}
Let $\quasi{\cdot}$ be a quasi-norm on $Z$, then the following equivalent expressions define a norm on Z. For $z\in Z$
\begin{align}
\norm{z}&=\inf \left\{\sum_{j=1}^{n}\quasi{z_{j}}: z=\sum_{j=1}^{n}z_{j} \right\}, \label{eq:norm} \\
&= \inf \{\lambda>0: z/\lambda \in \operatorname{conv}(B_{Z}) \}, \label{eq:gauge}\\
&= \inf\{\xi(z): \xi\in Z^{*}, \norm{\xi}\leq 1 \}. \label{eq:dnorm}
\end{align}
Moreover, for the quasi-norm defined by Eq.~\eqref{eq:qnorm} we have the following equivalence
\begin{equation}\label{eq:equiv}
\norm{(y,x)}\leq \quasi{(y,x)}_{F} \leq 2\norm{(y,x)}.
\end{equation}

\end{lemma}

\begin{proof} We show first that the first expression indeed defines a norm. Since $\quasi{\cdot}$ is a quasi-norm, the only property that we need to check is the triangle inequality. This can be verified by 

\begin{align*}
\norm{z_{1}+z_{2}}&=\inf \left\{\sum_{j=1}^{n}\quasi{w_{j}}: z_{1}+z_{2}=\sum_{j=1}^{n}w_{j}=\sum_{j=1}^{n_{1}}w_{j}+\sum_{j=1}^{n_{2}}w_{j} \right\}, \\
&\leq \inf \left\{\sum_{j=1}^{n_{1}}\quasi{w_{j}}: z_{1}=\sum_{j=1}^{n_{1}}w_{j} \right\}+\inf \left\{\sum_{j=1}^{n_{2}}\quasi{w_{j}}: z_{2}=\sum_{j=1}^{n_{2}}w_{j} \right\}, \\
&= \norm{z_{1}}+\norm{z_{2}},
\end{align*} 
as those are valid decompositions of $z_{1}+z_{2}$. We show now that Eq.~\eqref{eq:norm} and Eq.~\eqref{eq:gauge} are the same. Let $\alpha=\norm{z}$ be the infimum of Eq.~\eqref{eq:gauge}. Then there exist $m\in\mathbb{N}$, positive real numbers $(\lambda_{j})_{j=1}^{m}$, $\sum_{j=1}^{m}\lambda_{j}=1$ and $(z_{j})_{j=1}^{m}$ with quasi-norm one such $z=\alpha\sum_{j=1}^{m}\lambda_{j}z_{j}$. This is a valid decomposition of $z$ and $\sum_{j=1}^{m}\quasi{\alpha\lambda_{j}z_{j}}\leq \alpha$. On the other hand, let $z=\sum_{j=1}^{m}z_{j}$ be the decomposition that achieves the infimum in Eq.~\eqref{eq:norm} so that $\norm{z}=\sum_{j=1}^{m}\quasi{z_{j}}$. Then 
\begin{equation*}
 \frac{z}{\sum_{k=1}^{m}\quasi{z_{k}}}=\sum_{j=1}\left(\frac{\quasi{z_{j}}}{\sum_{k=1}^{m}\quasi{z_{k}}}\right)\frac{z_{j}}{\quasi{z_{j}}}\in \operatorname{conv}(B_{Z}).
\end{equation*}

The norm of $\xi\in Z^{*}$ can be computed as

\begin{equation*}
 \norm{\xi}=\sup_{z\in \operatorname{conv}(B_{Z})}|\xi(z)|= \sup_{z\in B_{Z}}|\xi(z)|=\sup\{ |\xi(z)| : \quasi{z}\leq 1 \},
\end{equation*}
as the supremum over a convex function is achieved at the extremal points. Thus the dual of the quasi-Banach space $Z$ and its Banach envelope coincide. Thus Eq.~\eqref{eq:dnorm} is just the usual expression in terms of the dual. We now compare the quasi-norm in Eq.~\eqref{eq:qnorm} with the norm of its envelope.

Since $\norm{z}$ is defined by the infimum of $\sum_{j}\quasi{z_{j}}$ over all the decompositions of $z$, Eq.~\eqref{eq:norm}, we immediately have the first inequality in Eq.~\eqref{eq:equiv}. For the second inequality let $(y,x)=\sum_{j}(y_{j},x_{j})$, then using Eq.~\eqref{eq:almostl}
\begin{align*}
\quasi{(y,x)}_{F}&=\frac{\norm{F(x)-y}_{Y}}{\delta}+\norm{x}_{X},\\
&=\frac{\norm{F(\sum_{j}x_{j})-\sum_{j}F(x_{j})+\sum_{j}F(x_{j})-\sum_{j}y_{j}}_{Y}}{\delta}+\norm{\sum_{j}x_{j}}_{X},\\
&\leq \frac{\norm{F(\sum_{j}x_{j})-\sum_{j}F(x_{j})}_{Y}}{\delta}+\sum_{j}\frac{\norm{F(x_{j})-y_{j}}_{Y}}{\delta}+\sum_{j}\norm{x_{j}}_{X}, \\
&\leq \frac{\delta\sum_{j}\norm{x_{j}}_{X}}{\delta}+\sum_{j}\frac{\norm{F(x_{j})-y_{j}}_{Y}}{\delta}+\sum_{j}\norm{x_{j}}_{X}, \\
&\leq 2\sum_{j}\quasi{(y_{j},x_{j})}_{F}.
\end{align*}
\end{proof}

Additionally, we can understand the resulting twisted sum $Z$ with norm as in Eq.~\eqref{eq:norm} as the space with unit ball~\cite{ELP75}

\begin{equation*}
 B_{Z}:=\operatorname{conv}\left(\{(y,0): \norm{y}_{Y}\leq 1\}\cup\{(F(x),x): \norm{x}_{X}\leq 1 \}\right).
\end{equation*}

We write $Z=Y\oplus_{F}X$ for the (Banach envelope of) twisted sum of $Y$ and $X$ generated by the almost-linear map $F:X\to Y$.

\section{Main result}

We are now ready to put all the pieces together and to make the connection explicitely between (co-)type constants and the linear stability of almost-linear maps. 
\begin{theorem}\label{th:k}
Let $F:X\to Y$ be an almost-linear map between finite dimensional real Banach spaces, i.e. $F$ is a real homogeneous map and there exist a $\delta>0$ such that for any finite sequence $(x_{i})_{i=1}^{m}\subset X, m\in\mathbb{N}$ and $\lambda\in\reals^{m}$, 

\begin{equation*}
\norm{\sum^{m}_{i=1}\lambda_{i}F(x_{i})-F\left(\sum^{m}_{i=1}\lambda_{i}x_{i}\right)}_{Y}\leq \delta \sum^{m}_{i=1}|\lambda_{i}|\norm{x_{i}}_{X}.
\end{equation*}
Let $Z=Y\oplus_{F}X$ be the respective twisted sum generated by this map. Then 
\begin{equation}
\inf_{H\in L(X,Y)}\;\sup_{x\in X}\frac{\norm{F(x)-H(x)}_{Y}}{\norm{x}_{X}} \leq 2\delta \min \{T_{2}(Z)C_{2}(X),1+T_{2}(Z^{*})C_{2}(Y^{*})\}, 
\end{equation}
where $T_{2}$ and $C_{2}$ are the type 2 and cotype 2 constants.
\end{theorem}

\begin{proof}[Proof of Theorem~\ref{th:k}]
We need the following important theorem of Maurey~\cite{M74} (see Theorem 7.4.4 in Ref.~\cite{AK00} for a modern proof).

\begin{theorem}(Maurey's Extension)\label{th:Maurey} 
Let $E$ be a Banach space and $S$ a closed subspace of $E$. Let $T_{2}(E)$ be either the Gaussian or Rademacher type 2 constant of $E$ and $C_{2}(S)$ either the Gaussian or Rademacher cotype 2 constant of $S$. Then there exist a projection $P:E\to S$ with 
\begin{equation*}
\norm{P}\leq T_{2}(E)C_{2}(S).
\end{equation*}
\end{theorem}

We remark that the norm $\norm{\cdot}$ in Theorem~\ref{th:Maurey} is the operator norm. This might seem odd at first sight as usually the projections are considered between Hilbert spaces and in that case they always have norm equal to one. This is no longer true when we leave the special world of Hilbert spaces and consider general Banach spaces. Maurey's theorem is proven by factorizing through a Hilbert space though. In a sense, the notions of type and cotype measure how far we are from the Hilbert space scenario. \\

Let $Z=Y\oplus_{F}X$ be the twisted sums of $Y$ and $X$ and consider the Banach envelope of $Z$. Let us denote by $Z$ as well the Banach envelope of $Z$. From Maurey's theorem we know there exist a projection $P:Z\to X$ such that $\norm{P}\leq T_{2}(Z)C_{2}(X)$. Since $P$ is a projection, it has the general form $P(y,x)=(y-H(x),0)$ where $H:X\to Y$ is a linear map.  Then using Eq.~\eqref{eq:equiv}
\begin{align*}
\norm{P}=\sup_{(y,x)\in Z}\frac{\norm{P(y,x)}}{\norm{(y,x)}} &\geq \sup_{x\in X} \frac{\norm{P(F(x),x)}}{\norm{(F(x),x)}} \\
&= \sup_{x\in X} \frac{\norm{(F(x)-H(x),0)}}{\norm{(F(x),x)}} \\
&\geq \sup_{x\in X} \frac{\quasi{(F(x)-H(x),0)}}{2\quasi{(F(x),x)}} \\
&= \sup_{x\in X} \frac{\norm{F(x)-H(x)}_{Y}}{2\delta\norm{x}_{X}} \\
&\geq  \inf_{H\in L(X,Y)} \sup_{x\in X} \frac{\norm{F(x)-H(x)}_{Y}}{2\delta\norm{x}_{X}}
\end{align*}

We can also consider a dual construction for a different bound. Let $Z^{*}$ be the dual of the twisted sum $Y\oplus_{F}X$. It is known~\cite{SC00} that the dual of $Z$ is isomorphic to $X^{*}\oplus_{F^{*}}Y^{*}$ where $F^{*}$ is in some sense the dual map of $F$ (see [\onlinecite{SC00}] for details). Since we are dealing with finite dimensional spaces, $Z^{**}$ can be identified with $Z$.  Let $Q:Z^{*}\to Y^{*}$ be the projection obtained by Maurey's extension theorem when applied to the Banach spaces $Z^{*}$ and $Y^{*}$. Let us consider the projection $\tilde{P}:Z\to X$ defined via $\tilde{P}:=\mathrm{id}-Q^{*}\pi$ where $\pi$ is the quotient map $\pi:Z\to X$, $\pi(y,x)=x$. Indeed let $\Omega:X\to Y$ be the linear map induced by $Q^{*}$. Then 
\begin{align*}
\tilde{P}(y,x)&=(y,x)-Q^{*}\pi (y,x), \\
&=(y,x)-Q^{*}x,\\
&=(y,x)-(\Omega(x),x)=(y-\Omega(x),0)\in X.
\end{align*}

Analogously as the previous calculation, we find 
\begin{equation*}
\inf_{\Omega\in L(X,Y)} \sup_{x\in X} \frac{\norm{F(x)-\Omega(x)}_{Y}}{\norm{x}_{X}} \leq 2\delta\norm{\tilde{P}} \leq 2\delta(1+\norm{Q}).
\end{equation*}

The final results then follows from the upper bound that Maurey's theorem provides on the norm of such projections.

\end{proof}

\section{Applications}

%
%
%
%
%

The following result gives an improvement on Theorem 2-(ii) in \cite{CW18}

\begin{theorem}[Linear Stability of Wigner's theorem]\label{th:stabwigner}
Let $f:\mathbb{P}(\comp^{d})\to\mathbb{P}(\comp^{d})$ be a function that satisfies 
\begin{equation}\label{eq:W}
\left| \langle f(x), f(y)\rangle - \langle x, y \rangle  \right|\leq \varepsilon \qquad \qquad \text{for all }x,y\in\mathbb{P}(\comp^{d}).
\end{equation}
Then there exist a universal constant $C$ and a linear map $H:\mathcal{H}_{d}\to \mathcal{H}_{d}$ such that for all $x\in\mathbb{P}(\comp^{d})$
\begin{equation*}
\norm{f(x)-H(x)}_{2}\leq (C\log_{2} d)^{\beta} \sqrt{d\varepsilon},
\end{equation*}
where $\beta=2+\tfrac{1}{2}\log_{2}\log_{2}2d$.
\end{theorem}

We call a map $f:\mathbb{P}(\comp^{d})\to\mathbb{P}(\comp^{d})$ which satisfies Eq.~\eqref{eq:W} an \textit{almost-symmetry}. In order to prove Theorem~\ref{th:stabwigner} we make use of the following lemmas (c.f. Theorem 1 in [\onlinecite{ELP75}]). First, we need the type constant of a twisted sum (cf. Lemma 16.6-7 in~[\onlinecite{BL00}])

\begin{lemma}\label{l:TStype}
Let $Z$ be the twisted sum of $Y$ and $X$, then
\begin{equation}\label{eq:typeZ}
T_{2,n^{2}}(Z)\leq T_{2,n}(Y)T_{2,n}(Z)+T_{2,n}(Y)T_{2,n}(X)+T_{2,n}(Z)T_{2,n}(X).
\end{equation}
\end{lemma}

The type 2 constant of a Banach space of dimension $d$ can be obtained from the type constant restricted to families of size $d(d+1)/2$ as stated by the following lemma. This result follows from a cone version of Caratheodory's theorem (see Lemma 6.1 in [\onlinecite{TLM77}]).

\begin{lemma}\label{l:finitetype}
Let $X$ be a $d-$dimensional Banach space. Then $T_{2,n}(X)=T_{2,d(d+1)/2}(X)$ and $C_{2,n}(X)=C_{2,d(d+1)/2}(X)$ for any $n\geq d(d+1)/2$.
\end{lemma}

\begin{proof}
The first step of the proof consist in extending the function $f$ to $F:S_{1}^{d}\to S_{1}^{d}$ such that $F|_{\mathbb{P}(\comp^{d})}=f$.  
We take $x$ in the unit sphere of $S_{1}^{d}$ and identify it with its antipodal point $-x$. We choose a fixed spectral decomposition for both elements, say $x=\sum_{j=1}^{d}\lambda_{j}x_{j}$, and define $F(x):=\sum_{j=1}^{d}\lambda_{j}f(x_{j})$. Then, we can extend $F$ homogeneously from the unit sphere to any $y\in S_{1}^{d}$ by multiplying $x$ or $-x$ with $\lambda\geq 0$ so that $\lambda x=y$ or $-\lambda x=y$. We call again this extension $F$. By construction, $F$ is a real homogeneous map. Note that this extension is not unique, but we do not need this here.\\

As proven in Lemma 2 of Ref.~\cite{CW18}, $F$ is an almost-linear map
\begin{equation*}
\norm{\sum^{m}_{i=1}\lambda_{i}F(x_{i})-F\left(\sum^{m}_{i=1}\lambda_{i}x_{i}\right)}_{2}\leq \delta \sum^{m}_{i=1}|\lambda_{i}|\norm{x_{i}}_{1},
\end{equation*}
with $\delta=2\sqrt{\varepsilon}$. If we use Theorem~\ref{th:k} with the twisted sum $Z:=S_{2}^{d}\oplus_{F}S_{1}^{d}$, we cannot obtain anything better than a linear dependence on $d$. However, we will be able to obtain a better dimension dependence if we consider a dual construction, namely with $Z^{*}:=S_{\infty}^{d}\oplus_{F^{*}}S_{2}^{d}$. For that matters we use Lemma~\ref{l:TStype} and Lemma~\ref{l:finitetype} in order to estimate the type 2 constant of $Z^{*}$. From Eq.~\eqref{eq:typeZ} and $T_{2}(S_{\infty})\leq \sqrt{4\log d}$, we obtain $T_{2,n^{2}}(Z^{*})\leq 2\sqrt{8\log_{2} d}\; T_{2,n}(Z^{*})$ for all $n\in\mathbb{N}$. It is known that for a general Banach space $E$, $T_{2}(E)\leq \sqrt{\operatorname{dim}(E)}$  (Proposition 12.3 in \cite{TJ89}). Thus for all 2-dimensional subspaces of $Z$ the type is less than $\sqrt{2}$ and $T_{2,2}(Z)\leq \sqrt{2}$ (this can be alternatively derived from a classical result of John and the relation between the Banach-Mazur distance and type 2 constants). It follows from induction that 

\begin{equation*}
 T_{2,2^{2^{k}}}(Z^{*})\leq (2\sqrt{8\log_{2} d})^{k}\sqrt{2},
\end{equation*}
 which in turns implies 

\begin{equation*}
T_{2,n}(Z^{*})\leq \sqrt{2}(\log_{2}n)(8\log_{2}d)^{\frac{\log_{2}\log_{2}n}{2}}.
\end{equation*}
The dimension of the real vector space of Hermitian matrices $\mathcal{H}_{d}$ is $d^{2}$. Therefore we obtain from Lemma~\ref{l:finitetype} with $n=2d^{4}$

\begin{equation*}
T_{2}(Z^{*})\leq 2(8\log_{2}d)^{2+\tfrac{\log_{2}\log_{2}2d}{2}}.
\end{equation*}
It follows from Theorem~\ref{th:k} and $C_{2}(S_{\infty}^{d})\leq\sqrt{d}$, that there exist a linear map $H:S_{1}^{d}\to S_{2}^{d}$ such that
\begin{equation*}
\sup_{x\in\mathbb{P}(\comp^{d})}\norm{f(x)-H(x)}_{2} \leq 4(8\log_{2} d)^{2+\tfrac{\log_{2}\log_{2}2d}{2}} \sqrt{d\varepsilon}. 
\end{equation*}

\end{proof}

%

The following proposition is essentially due to Kalton. It can be shown using Theorem 2.2 in \cite{K91} as $S_{2}^{d}$ and $\reals^{d^{2}}$ are isomorphic Hilbert-spaces. We present here a proof using the notions of (co-)type and Theorem~\ref{th:k}.

\begin{proposition}[Stability of Global Symmetries]\label{cor:GlobalW}
Let $f:B_{S_{2}^{d}}\to S_{2}^{d}$ be a continuous function that satisfies 
\begin{equation}\label{eq:GlobalW}
\left| \langle f(x), f(y)\rangle - \langle x, y \rangle  \right|\leq \varepsilon \qquad \qquad \text{for all }x,y\in B_{S_{2}^{d}}.
\end{equation}
Then there exist a linear map $H:S_{2}^{d}\to S_{2}^{d}$ and an absolute constant $C$ such that for all $X\in B_{S_{2}^{d}}$
\begin{equation*}
\norm{f(x)-H(x)}_{2}\leq C \sqrt{\varepsilon} \log_{2} d .
\end{equation*}
\end{proposition}

\begin{proof}[Proof of Proposition~\ref{cor:GlobalW}]
The first step consist of showing that the function $f$ can be extended to a continuous homogeneous function on the whole space without paying much. 

\begin{lemma}\label{l:qlinearext}
Let $f:B_{S^{d}_{2}}\to S_{2}^{d}$ be a continuous function that satisfies 
\begin{equation*}
\left| \langle f(x), f(y)\rangle - \langle x, y \rangle  \right|\leq \varepsilon \qquad \qquad \text{for all }x,y\in B_{S^{d}_{2}}.
\end{equation*}
Then there exists a continuous and homogeneous function $F:S_{2}^{d}\to S_{2}^{d}$ such that
\begin{equation}\label{eq:quasi-affine}
\norm{  \sum_{j=1}^{n}F\left(x_{j}\right)-F\left(\sum_{j=1}^{n}x_{j}\right)}_{2}\leq  4\sqrt{\varepsilon}\;\sum_{j=1}^{n}\norm{x_{j}}_{2} \qquad \qquad \text{for all }x_{j}\in S_{2}^{d},
\end{equation}

and

\begin{equation}\label{eq:nearF}
\sup_{X\in B_{S^{d}_{2}}}\norm{f(x)-F(x)}_{2}\leq 3\sqrt{\varepsilon}.
\end{equation}

\end{lemma}

\begin{proof}[Proof of Lemma~\ref{l:qlinearext}]
Let us extend $f$ to $F:S_{2}^{d}\to S_{2}^{d}$ where
\begin{equation*}
F(x):=\norm{x}_{2}\left(f\left(\frac{x}{2\norm{x}_{2}}\right)-f\left(-\frac{x}{2\norm{x}_{2}}\right) \right).
\end{equation*}
This function is homogeneous, i.e. $F(\lambda x)=\lambda F(x)$ for all $\lambda\in\reals$, and continuous as $f$ and $\norm{\cdot}$ are also continuous. Using Eq.~\eqref{eq:GlobalW} and the triangle inequality, we obtain the new almost-symmetry condition

\begin{equation}\label{eq:Wquasi}
 |\langle F(x),F(y) \rangle - \langle x, y \rangle|\leq 4\varepsilon\norm{x}_{2}\norm{y}_{2}.
\end{equation}

Hence, for any $z\in S_{2}^{d}$
\begin{align*}
 \left|\left\langle \sum_{j=1}^{n}F(x_{j})-F\left(\sum_{j=1}^{n}x_{j}\right),F(z) \right\rangle\right| &=  \left|\left\langle \sum_{j=1}^{n}F(x_{j})-F\left(\sum_{j=1}^{n}x_{j}\right),F(z) \right\rangle - \left\langle \sum_{j=1}^{n}x_{j}-\sum_{j=1}^{n}x_{j}, z \right\rangle\right| \\ 
&\leq 8\varepsilon\;\sum_{j=1}^{n}\norm{x_{j}}_{2}\norm{z}_{2}.
\end{align*}
Therefore, from the linearity of the inner product we obtain Eq.~\eqref{eq:quasi-affine}.
Finally, we show that $f$ and $F$ are $\sqrt{\varepsilon}-$close. From Eq.~\eqref{eq:GlobalW} 

\begin{align*}
 &|\langle F(x),f(x) \rangle -\langle x, x\rangle | \\
 &=|\norm{x}_{2}\left( \langle f\left(\frac{x}{2\norm{x}}\right),f(x) \rangle - \frac{\norm{x}_{2}}{2} \right)-\norm{x}_{2}\left( \langle f\left(\frac{-x}{2\norm{x}}\right),f(x) \rangle +  \frac{\norm{x}_{2}}{2} \right)|\\
 &\leq 2\varepsilon\norm{x}_{2}.
\end{align*}
Thus, with Eq.~\eqref{eq:Wquasi} we have
\begin{align*}
\norm{F(x)-f(x)}^{2}_{2}&=\norm{F(x)}^{2}_{2}-\norm{x}^{2}_{2}-2 \operatorname{Re}\left(\langle F(x),f(x) \rangle -\langle x, x\rangle\right)+\norm{f(x)}^{2}_{2}-\norm{x}^{2}_{2}\\
&\leq 4\varepsilon\norm{x}^{2}_{2}+4\varepsilon\norm{x}_{2}+\varepsilon
\end{align*}
which is less than $9\varepsilon$ for all $x\in B_{S_{2}^{d}}$.

\end{proof}



We consider now the twisted sum $Z=S_{2}^{d}\oplus_{F} S_{2}^{d}$ generated by the almost-linear map $F$. Before applying Theorem~\ref{th:k} we estimate the type 2 constant of $Z$. Since $S^{d}_{2}$ is a Hilbert space, it has type 2 constant equal to one and we obtain from Lemma~\ref{l:TStype} that 
\begin{equation}
T_{2,n^{2}}(Z)\leq 1+2T_{2,n}(Z)\qquad\text{for all }n. 
\end{equation}

As in the proof of Theorem~\ref{th:stabwigner}, all 2-dimensional subspaces of $Z$ have type less than $\sqrt{2}$ and $T_{2,2}(Z)\leq \sqrt{2}$. It follows from induction that for $n\geq 3$,
 \begin{equation*}
 T_{2,n}(Z)\leq 2(1+\sqrt{2})\log_{2} n.
 \end{equation*}
 Hence, from Lemma~\ref{l:finitetype} with $n=4d^{2}$
 \begin{equation*}
 T_{2}(Z)\leq 4(1+\sqrt{2})\log_{2} 2d.
 \end{equation*}
Accordingly, from $C_{2}(S_{2}^{d})=1$ and Theorem~\ref{th:k} there exist a linear map $H:S_{2}^{d}\to S_{2}^{d}$ such that for all $x\in B_{S_{2}^{d}}$
 \begin{equation}
 \norm{F(x)-H(x)}_{2} \leq 32(1+\sqrt{2})\log_{2}(2d)\sqrt{\varepsilon}. 
 \end{equation}
Finally, from Eq.~\eqref{eq:nearF} and the triangle inequality we obtain
\begin{align*}
\sup_{x\in B_{S_{2}^{d}}} \norm{f(x)-H(x)}_{2}&\leq \sup_{x\in B_{S_{2}^{d}}} \norm{f(x)-F(x)}_{2}+\norm{F(x)-H(x)}_{2} \\
&\leq 79\sqrt{\varepsilon}\;(1+\log_{2}d).
\end{align*} 

\end{proof}
\section{Outlook}

Using Theorem~\ref{th:k} we are able to improve --up to some logarithmic factors-- the upper bound on the dimension dependence of the linear stability of Wigner's theorem from $d$ to $\sqrt{d}$. There is still room for an exponential improvement in the dimension as the lower bound is of order $\log d$ (see the discussion section of~\cite{CW18}). Even if we were able to extend the almost-symmetry $f$ to an almost-linear map $F:S_{1}^{d}\to S_{1}^{d}$ with $\delta$ independent of $d$, we would still get from Theorem~\ref{th:k} an upper bound of order $\sqrt{d}$. This is just a consequence of how the type and cotype constants of $S_{1}^{d}$ and $S_{\infty}^{d}$ behave. There is a trade-off in Theorem~\ref{th:k} between the type constant for individual spaces and the type constant of their twisted sum. It can be seen from Table~\ref{tab:type} and Lemma~\ref{l:TStype} that the best bound that can be obtained from Theorem~\ref{th:k} is in the case that $X$ and $Y$ are Hilbert spaces. This is the case of proposition~\ref{cor:GlobalW}. However, there the almost-symmetry condition holds for the entire Hilbert-Schmidt unit ball, while in Wigner's theorem the almost-symmetry condition is required to hold only for the non-linear space of normalized hermitian rank-one projections.\\

%


\emph{Acknowledgments:} The author would like to thank Marius Junge, Willian Corr\^ea and Cambyse Rouz\'e for valuable discussions. Furthermore, the author wishes to thank the Institut Henri Poincar\'e in Paris and the organizers of the trimester on ``Analysis in Quantum Information Theory" (IHP17). 




\end{document}